\documentclass[american,aps,pra,reprint, superscriptaddress]{revtex4-1}
\usepackage[unicode=true,pdfusetitle, bookmarks=true,bookmarksnumbered=false,bookmarksopen=false, breaklinks=false,pdfborder={0 0 0},backref=false,colorlinks=false] {hyperref}
\hypersetup{ colorlinks,linkcolor=myurlcolor,citecolor=myurlcolor,urlcolor=myurlcolor}
\usepackage{graphics,epstopdf,graphicx, amsthm, amsmath, amssymb, times, braket, colortbl, color, bm, framed, cleveref, mathrsfs}
\usepackage{textcase}
\usepackage[up]{subfigure}
\usepackage{tikz}
\definecolor{myurlcolor}{rgb}{0,0,0.7}

\newcommand{\tinyspace}{\mspace{1mu}}

\newcommand{\proj}[1]{| #1\rangle\!\langle #1 |}

\newcommand{\iinner}[2]{\langle #1 | #2\rangle}

\newcommand{\Pa}[1]{\left(#1\right)}
\newcommand{\Br}[1]{\left[#1\right]}
\newcommand{\abs}[1]{\left\lvert\tinyspace #1 \tinyspace\right\rvert}
\newcommand{\norm}[1]{\left\lVert #1 \right\rVert}

\theoremstyle{plain}
\newtheorem{thm}{Theorem}
\newtheorem{lem}[thm]{Lemma}
\newtheorem{prop}[thm]{Proposition}
\newtheorem{cor}[thm]{Corollary}

\newcommand*{\myproofname}{Proof}

\def\ot{\otimes}
\def\complex{\mathbb{C}}
\def\real{\mathbb{R}}

\def\cI{\mathcal{I}}

\def \dif {\mathrm{d}}

\def\rS{\mathrm{S}}

\makeatother

\begin{document}

 \author{Kaifeng Bu}
 \email{bkf@zju.edn.cn}
 \affiliation{School of Mathematical Sciences, Zhejiang University, Hangzhou 310027, PR~China}
\author{Swati}
 \email{swati.ssingh03@gmail.com}
 \affiliation{Amity Institute of Applied Sciences, Amity University, Noida, 211019, India}
 \author{Uttam Singh}
 \email{uttamsingh@hri.res.in}
 \affiliation{Harish-Chandra Research Institute, Allahabad, 211019, India}
  \author{Junde Wu}
 \email{wjd@zju.edn.cn}
 \affiliation{School of Mathematical Sciences, Zhejiang University, Hangzhou 310027, PR~China}

\title{Coherence breaking channels and coherence sudden death}

\begin{abstract}
Quantum noise is ubiquitous to quantum systems as they incessantly interact with their surroundings and results in degrading useful resources such as coherence for single quantum systems and quantum correlations for multipartite systems. Given the importance of these resources in various quantum information processing protocols, it is of utmost importance to characterize how deteriorating is a particular noise scenario (quantum channel) in reference to a certain resource? Here we develop a theory of coherence breaking channels for single quantum systems. Any quantum channel on a single quantum system will be called a coherence breaking channel if it is an incoherent channel and maps any state to an incoherent state. We explicitly and exhaustively characterize these coherence breaking channels. Moreover, we define the coherence breaking indices for incoherent quantum channels and present various examples to elucidate this concept. We further introduce the concept of coherence sudden death under noisy evolutions and make an explicit connection of the phenomenon of coherence sudden death with the coherence breaking channels and the coherence breaking indices together with various suggestive examples. Furthermore, for higher dimensional Hilbert spaces, we establish the typicality of the dynamics of coherence under any incoherent quantum channel exploiting the concentration of measure phenomenon.
\end{abstract}

\maketitle

\section{Introduction}
Quantum coherence and entanglement have been corner stones of the quantum information theory both with discrete quantum systems and continuous variable quantum systems \cite{Schrodinger1935, Sudarshan63}. Though the concept of quantum entanglement requires quantum systems with at least two parties, coherence can be defined for a single quantum system \cite{Ivan2013}. However, these two notions are closely related. For example, entanglement is nonlocal nonclassicality for continuous variable quantum systems \cite{Ivan2013} while for discrete quantum one can convert coherence into entanglement \cite{Alex15}. Recent developments in the fields of quantum biology \cite{Plenio2008,Lloyd2011,Levi14} and the theory of quantum thermodynamics \cite{Aspuru13,Rudolph214,Lostaglio2015,Brandao2015,Narasimhachar2015, Piotr2015, Misra2016} urge for better understanding of quantum coherence and therefore, the resource theories of quantum coherence are developed in recent years \cite{Baumgratz2014, Yadin2015, Aberg14, Bai2015, Bromley2015, Cheng2015, Chitambar2016, ChitambarA2015, Chitambar2015, Du2015, Girolami14, Hu2016, Mani2015, Marvian14, MondalA2015, Napoli2016, Peng2015, Piani2016, Swapan2016, Fan15, Uttam2015, UttamA2015, UttamS2015, Singh2016, Zhang2015, Buc2015, BuD2016, Bu2015, AlexB2015, StreltsovA2015, Xi2015, Yao2015, Winter2016}.

Quantum systems with coherence are resourceful in quantum thermodynamics and quantum biology, however, the resourcefulness of a quantum system gets degraded over the time due to its uncontrollable constant interactions with the surrounding environment. These interactions of a quantum system with the surrounding environments give rise to various quantum noise models (noisy quantum channels). The characterization of these noisy channels and their effect on various physical resources are of huge practical value and therefore various special kinds of noisy channels have been considered \cite{Nielsen10}. For example, in the context of entanglement as a resource, entanglement breaking channels have been characterized completely \cite{Horodecki2003,Ruskai2003} and similarly, nonclassicality-breaking channels \cite{Ivan2013,Sabapathy2015} have been analysed in detail in the context of nonclassicality as a resource for continuous variable quantum systems. The entanglement breaking channels and the nonclassicality breaking channels make any input state separable or classical, respectively. The characterizations of the entanglement breaking channels  and the nonclassicality breaking channels are given in Refs. \cite{Horodecki2003} and \cite{Sabapathy2015}, respectively.

In this work, we focus on the effect of noisy quantum channels on coherence of single quantum systems. To address this problem, we first define coherence breaking channels (CBCs) which output only incoherent states for every input state and completely characterize the form of such channels. We then introduce selective coherence breaking channels (SCBCs) and prove that the characterization of these two classes of channels is equivalent. We also comment on the possible connection between the entanglement breaking channels and the coherence breaking channels. Moreover, we provide the inter-relation between the coherence breaking channels and other relevant incoherent  operations such as the strictly incoherent operations (SIOs) \cite{Yadin2015} and the dephasing covariant operations (DIOs) \cite{Chitambar2016}. We then define the coherence breaking indices of incoherent quantum channels. The coherence breaking index $n(\Phi)$ of an incoherent quantum channel $\Phi$ is the minimum number of iterations of $\Phi$ that are required in order to make $\Phi$ a coherence breaking channel, i.e., $n(\Phi) = \min\{ n: \Phi^n \in \mathcal{S}_{\mathrm{cbc}}\}$. Here, $\mathcal{S}_{\mathrm{cbc}}$ is the set of all coherence breaking channels. Further, we introduce the notion of coherence sudden death during incoherent evolutions of a quantum system and importantly characterize the coherence sudden death using the coherence breaking channels and the coherence breaking indices. We provide various suggestive examples to elucidate the concept of the coherence sudden death. Moreover, we show the concentration effect for the coherence of evolved state under any incoherent quantum channel starting from a random pure state chosen according to the Haar distribution for higher dimensional Hilbert spaces. This result is based on the extremely powerful result known as the L\'evy's lemma and establishes the universality of the coherence dynamics under incoherent evolutions.

This work is organized as follows. We start with a brief introduction of quantum channels and the resource theory of coherence in Sec. \ref{sec:prelims}. In Sec. \ref{sec:cbt} we define two kinds of coherence breaking channels and give  the characterizations of both the channels. We link the coherence breaking channels with SIOs and DIOs in Sec. \ref{sec:rel}. Sec. \ref{sec:dis} is devoted to the coherence breaking indices of incoherent quantum channels. In Sec. \ref{sec:coh-sud} we discuss the notion of coherence sudden death and its relation to the coherence breaking channels and  the coherence breaking indices. Moreover, we establish the universality (typicality) of the coherence dynamics in higher dimensional Hilbert spaces. Finally, we conclude in Sec. \ref{sec:con} with a brief overview of the results obtained in this work.

\section{Preliminaries}
\label{sec:prelims}
Here we present the relevant basic tools and concepts that are required for presenting our main results.

\noindent
{\it Quantum channels.--} A quantum channel is a linear completely positive and trace preserving (CPTP) map which appears naturally in the description of open quantum systems and plays an important role in quantum information theory \cite{Nielsen10}. According to the Kraus representation theorem, any linear map $\Phi$ is a CPTP map if and only if it can be represented by a set of Kraus operators $\set{K_i(\Phi)}_{i=1}^{N}$ as follows.
\begin{align}
\Phi[\rho]:=\sum_{i=1}^{N}K_i(\Phi)\rho K_i^\dagger(\Phi)
\end{align}
with $\sum_nK^\dag_nK_n=\mathbb{I}$. In this work we will require another very useful characterization of quantum channels which is the Choi-Jamio{\l}kowski isomorphism \cite{Jamiolkowski1972, Aolita2015}.

\noindent
{\it Choi-Jamio{\l}kowski isomorphism.--}Qualitatively, it states that a channel acting on a single party $S$ and the corresponding bipartite mixed state acting on $SA$ are informationally equivalent. Here $A$ is some ancilla system with the same dimension $d$ as of the system. More precisely, Choi-Jamio{\l}kowski isomorphism states that a channel $\Phi$ acting on a single quantum system is completely positive if and only if $\Phi\otimes \mathbb{I}$ applied to the maximally entangled state $\ket{\beta}^{SA}=\frac{1}{\sqrt{d}}\sum_{i=0}^{d-1}\ket{ii}^{SA}$ yields a positive semidefinite operator $\rho_\Phi$, i.e., $\rho_\Phi=\Phi\otimes \mathbb{I}[\ket{\beta}\bra{\beta}^{SA}]\geq 0$ with $\mathrm{Tr}_S[\rho_\Phi]=\frac{1}{d}$. Conversely, the isomorphism implies that for every positive semidefinite operator $\rho^{SA}$ there exists a unique channel $\Phi_\rho$ acting on $S$ such that $\rho^{SA} = \Phi_\rho\otimes\mathbb{I}[\ket{\beta}\bra{\beta}^{SA}]$.

\noindent
{\it Resource theories of quantum coherence.--} A resource theory comprises two basic elements: one is the set of allowed (free) operations and other being the set of allowed (free) states. The set of allowed operations is governed by the physical situations at hand. For example, in the resource theory of entanglement the allowed operations are the local operations and classical communication (LOCC) as it is not possible to implement global operations on two parties that are separated and located far away from each other. Similarly, the allowed operations in the other known resource theories such as the resource theory of thermodynamics and the resource theory of reference frames are obtained based on the relevant physical situations. However, there is still no general consensus on the set of allowed operations in the resource theory of coherence and we have resource theories of coherence based on incoherent operations and symmetric operations \cite{Gour2008, Marvian14, Baumgratz2014, AlexB2015, Chitambar2016}.

In this work, we consider the resource theory of coherence based on incoherent operations \cite{Baumgratz2014}. It is important to note here that the measures of coherence obtained using this resource theory are proved to be operationally meaningful \cite{UttamA2015, Winter2016} and play a crucial role in establishing quantitatively the emergence of entanglement from coherence \cite{Alex15}. The measures of coherence as obtained in this resource include $l_1$ norm of coherence and the relative entropy of coherence. For a $d$ dimensional quantum system in a state $\rho$ and a fixed reference basis $\{\ket{i}\}$, the $l_1$ norm of coherence $C_{l_1}(\rho)$ and the relative entropy of coherence $C_{r}(\rho)$ are respectively, defined as
\begin{align}
&C_{l_1}(\rho) = \sum_{\substack{{i,j=0}\\{i\neq j}}}^{d-1} |\bra{i}\rho\ket{j}| \mathrm{~and}\\
&C_r(\rho)=S(\rho^{(d)})-S(\rho),
\end{align}
where $\rS(\rho)=-\mathrm{Tr}{[\rho\log\rho]}$ is the von Neumann entropy of $\rho$ and $\rho^{(d)}$ is the diagonal part of $\rho$ in basis $\{\ket{i}\}$. We emphasize that the notion of coherence is intrinsically basis dependent and we only consider quantum systems with finite dimensional Hilbert spaces.

Given a fixed reference basis, say $\{\ket{i}\}$, any state which is diagonal in the reference basis is called an incoherent state. Let $\mathcal{I}$ be the set of all incoherent states. Then, an operation $\Phi$ is called an incoherent operation (IO) if the set of Kraus operators $\{K_n\}$ of $\Phi$ is such that $K_n\mathcal{I} K^{\dag}_n\subset \mathcal{I}$ for each $n$.

\section{Coherence breaking channels}
\label{sec:cbt}
We define coherence breaking channels and provide their exhaustive characterization in this section.
\subsection{Selective coherence breaking channels (SCBCs)}
A quantum channel $\Phi$ with Kraus operators $\set{K_n}$ is called a selective coherence breaking channel (SCBC) if $\Phi$ is an incoherent channel and for any state $\rho$, $K_n\rho K^\dag_n$ is an incoherent state for any $n$. Let the set of all SCBCs be denoted by $\mathcal{S}_{\mathrm{scbc}}$. It is easy to see that $\mathcal{S}_{\mathrm{scbc}}$ is convex. The theorem below characterizes SCBCs.
\begin{thm}
\label{th:scbt}
The following statements are equivalent:

\smallskip
\noindent
(i) A quantum channel $\Phi$ with Kraus decomposition $\{K_n\}$ is a selective coherence breaking channel.

\smallskip
\noindent
(ii) For any maximally coherent state $\ket{\psi}$, $K_n\proj{\psi}K^\dag_n\in \cI$ for any n.

\smallskip
\noindent
(iii) The Kraus operators $\set{K_n}$ of $\Phi$ can be written as $K_n=\ket{i_n}\bra{\phi_n}$ for any $n$, where $\ket{i_n}$ is an element of the reference basis $\set{\ket{i}}^{d-1}_{i=0}$ and $\sum_n\proj{\phi_n}=\mathbb{I}$.

\smallskip
\noindent
(iv) The action of $\Phi$ is given as $\Phi(\rho)=\sum_{i}\proj{i}\mathrm{Tr}(\rho F_i)$, where $\{F_i\}$ is a  set of positive semi-definite operators and $\sum_i F_i=\mathbb{I}$.
\end{thm}

\begin{proof}
The implication $(i)\Rightarrow (ii)$: It follows directly from the definition of $\mathcal{S}_{\mathrm{scbc}}$. The implication $(ii)\Rightarrow (iii)$: Suppose  for any maximally coherent pure state $\ket{\psi}$, $K_n\proj{\psi}K^\dag_n\in \cI$ for any $n$. Since $K_n$ can be written as $K_n=\sum_i\ket{i}\bra{\phi^n_i}$ and $K_n\ket{\psi}$ is incoherent, then by Lemma \ref{lem:noze} (see \cref{append:lemmas}), there exists only one index $i$ such that $\ket{\phi^n_i}\neq 0$, that is $K_n$ can be written as $K_n=\ket{i_n}\bra{\phi_n}$. As $\sum_nK^\dag_nK_n=\mathbb{I}$, then we have $\sum_n\proj{\phi_n}=\mathbb{I}$. The implication $(iii)\Rightarrow (iv)$: Since $K_n=\ket{i_n}\bra{\phi_n}$ for any $n$, then $\Phi(\rho)=\sum_nK_n \rho K^\dag_n=\sum_n\ket{i_n}\bra{\phi_n}\rho\ket{\phi_n}\bra{i_n}=\sum_i\proj{i}\mathrm{Tr}(\rho F_i)$, where $F_i$ is the sum of some (unnormalized)  pure states $\proj{\phi_n}$. As  $\sum_n\proj{\phi_n}=\mathbb{I}$, then we have $\sum_iF_i=\mathbb{I}$. The implication $(iv)\Rightarrow (i)$: As each $F_i$ can be written as $F_i=\sum_k\lambda_{ik}\proj{\phi^i_{k}}$, then $\Phi$ has the  Kraus representation $\set{K_{ik}}$ with $K_{ik}=\sqrt{\lambda_{ik}}\ket{i}\bra{\phi^i_{k}}$. Then, $\Phi$ is an incoherent operation and for any state $\rho$, $K_{ik}\rho K^\dag_{ik}$ is incoherent.
\end{proof}

It is important to note that the composition of any incoherent channel with a selective coherence breaking channel is again a selective coherence breaking channel, i.e., if $\Phi\in \mathcal{S}_{\mathrm{scbc}}$, then for any incoherent operation $\Omega$, $\Omega\circ\Phi$ and $\Phi\circ\Omega$ also belong to the set $\mathcal{S}_{\mathrm{scbc}}$. Let us recall that if a quantum channel $\Phi$ can be written in the following form
\begin{eqnarray}
\label{eq:Hov}
\Phi(\rho)=\sum_k Q_k \mathrm{Tr}(\rho P_k),
\end{eqnarray}
where each $Q_k$ is a density matrix and the $\set{P_k}$ forms a positive operator valued measure (POVM), then we say that the quantum channel $\Phi$ has the Holevo form \cite{Holevo1999}. Moreover, if each density matrix $Q_k=\proj{k}$ is a one-dimensional projection and $\sum_kP_k=\mathbb{I}$, then $\Phi$ is called a \emph{quantum-classical} (QC) channel. It is easy to see that a selective coherence breaking channel is a QC channel (see Theorem \ref{th:scbt}).

\subsection{Coherence breaking channels (CBCs)}
In the definition of SCBCs, we required that for any Kraus operator $K_n$ of a quantum channel, $K_n\rho K^\dag_n$ is an incoherent state for any state $\rho$. However, we can also define another kind of noisy channel, namely, the coherence breaking channel which requires weaker constraints on Kraus elements compared to the case of SCBCs. A quantum channel $\Phi$ is called a coherence breaking channel (CBC) if for any incoherent channel $\Phi$, $\Phi(\rho)$ is an incoherent state for any state $\rho$. Let the set of all CBCs be denoted by $\mathcal{S}_{\mathrm{cbc}}$. It can be seen easily that $\mathcal{S}_{\mathrm{cbc}}$ is convex and $\mathcal{S}_{\mathrm{scbc}}\subset \mathcal{S}_{\mathrm{cbc}}$. To characterize CBCs, let us start from the simplest case, namely, the quantum channels on qubit states. We know that any qubit state $\rho$ can be written as
\begin{eqnarray*}
\rho=\frac{\mathbb{I}+\vec{r}\cdot \vec{\sigma}}{2},
\end{eqnarray*}
where $|\vec{r}|\leq 1$. The action of a qubit quantum channel $\Phi$ is
completely characterized by a $3\times 3$ real matrix $M$ and
a $3$-dimensional vector $\vec{n}$ such that
\begin{eqnarray*}
\Phi\left(\frac{\mathbb{I}+\vec{r}\cdot \vec{\sigma}}{2}\right)
=\frac{\mathbb{I}+\left(M\vec{r}+\vec{n}\right)\cdot \vec{\sigma}}{2}.
\end{eqnarray*}
Now we can also use $(M, \vec{n})$ to denote a qubit quantum channel. Therefore, we can easily obtain the characterization of a qubit CBC as follows.

\begin{prop}
In qubit case, a quantum channel $\Phi$ represented by $(M,\vec{n})$ is coherence breaking  {\it if and only if} $M$ and $\vec{n})$ have the following forms:
\begin{equation*}
M= \Br{\begin{array}{ccc}
0 & 0 & 0  \\
0 & 0 & 0\\
M_{31} & M_{32} & M_{33}
\end{array}} \mathrm{~and~}
\vec{n}=\Br{\begin{array}{ccc}
0 \\
0 \\
n_z
\end{array}}.
\end{equation*}

\end{prop}
\begin{proof}
Qubit channel $\Phi$ is coherence breaking if and only if the first and second components of the vector $M\vec{r}+\vec{n}$
are $0$ for any vector $\vec{r}$, which means $M\vec{r}+\vec{n}=(0, 0, *)^T$. Thus, matrix $M$ and vector $\vec{n}$ must be of the following form:
\begin{equation*}
M= \Br{\begin{array}{ccc}
0 & 0 & 0  \\
0 & 0 & 0\\
M_{31} & M_{32} & M_{33}
\end{array}}\mathrm{~and~}
\vec{n}=\Br{\begin{array}{ccc}
0 \\
0 \\
n_z
\end{array}}.
\end{equation*}
Here we have not considered other restrictions on $\Phi$ like the complete positivity. In fact, the necessary and sufficient conditions for $(M,\vec{n})$ to be a CPTP map can be found in Refs. \cite{Ruskai2002,King2001}.
\end{proof}

In general case, the characterization of CBCs is given by the following theorem.

\begin{thm}\label{thm:cbt}
The following statements are equivalent.

\smallskip
\noindent
(i) $\Phi\in \mathcal{S}_{\mathrm{cbc}}$.

\smallskip
\noindent
(ii) $\Phi$ is an incoherent channel and for any maximally coherent state $\ket{\psi}$, $\Phi(\proj{\psi})\in \cI$.

\smallskip
\noindent
(iii) $\Phi$ is an incoherent channel and $\Phi(\ket{i}\bra{j})$ is diagonal for any two incoherent basis states $\ket{i}$ and $\ket{j}$.

\smallskip
\noindent
(iv) The action of $\Phi$ is given as $\Phi(\rho)=\sum_{i}\proj{i}\mathrm{Tr}(\rho F_i)$, where $\{F_i\}$  a set of positive semi-definite operators and $\sum_i F_i=\mathbb{I}$.
\end{thm}

\begin{proof}
The implication $(i)\Rightarrow (ii)$: This follows from the definition of $\mathcal{S}_{\mathrm{cbc}}$. The implication $(ii)\Rightarrow (iii)$:  Any maximally coherent pure state $\ket{\psi}$ in $d$ dimensional Hilbert space can be written as $\ket{\psi}=\frac{1}{\sqrt{d}}\sum^{d-1}_{j=0}e^{i\theta_j}\ket{j}$, where $\set{\ket{j}}^{d-1}_{j=0}$ is the reference basis. Then
\begin{eqnarray*}
\Phi(\proj{\psi})&=&\frac{1}{d}\sum_{i,j}e^{i(\theta_i-\theta_j)}\Phi(\ket{i}\bra{j})\\
&=&\frac{1}{d}\left[\sum^{d-1}_{i=0}\Phi(\proj{i})+\sum^{d-1}_{\substack{i,j=0\\ i\neq j}}e^{i(\theta_i-\theta_j)}\Phi(\ket{i}\bra{j})\right].
\end{eqnarray*}
For any $i$ and $j$, $\Phi(\ket{i}\bra{j})$ can be written as a matrix $[\Phi^{i,j}_{u,v}]\in\complex^{d\times d}$ in the reference basis. Since $\Phi$ is an incoherent channel, then $\Phi(\proj{i})$ is an incoherent state which means that $[\Phi^{i,i}_{uv}]$ is diagonal for any $i$. Moreover, $\Phi(\proj{\psi})$ is diagonal, thus $\sum^{d-1}_{\substack{i,j=0\\i\neq j}}e^{i(\theta_i-\theta_j)}[\Phi^{i,j}_{uv}]$ is diagonal. That is, for any fixed $u$ and $v$ with $u\neq v$, $
\sum^{d-1}_{\substack{i,j=0\\i\neq j}}e^{i(\theta_i-\theta_j)}\Phi^{i,j}_{uv}=0$
for any $(\theta_0, \ldots, \theta_{d-1})$. From Lemma \ref{lem:equ0} (see \cref{append:lemmas}),  $\Phi^{i,j}_{uv}=0$ for any $i\neq j$ with any fixed $u\neq v$. Thus, $[\Phi^{i,j}_{uv}]$ is diagonal, i.e., $\Phi(\ket{i}\bra{j})$ is diagonal for any $i$ and $j$. 

The implication $(iii)\Rightarrow (iv)$: Suppose that $\Phi$ is an incoherent channel and $\Phi(\ket{i}\bra{j})$ is diagonal for any $0\leq i,j\leq d-1$. Then $\Phi(\rho)$ is an incoherent state for any state $\rho$.  First, we prove that $\Phi$ is an entanglement breaking channel. Based on Ref. \cite{Horodecki2003}, we only need to prove $\mathbb{I}\otimes \Phi(\proj{\beta})$ is a separable state, where $\ket{\beta}=\frac{1}{\sqrt{d}}\sum^{d-1}_{i=0}\ket{ii}$. Without loss of generality, $
\Phi(\ket{i}\bra{j})=\sum_kd^{(i, j)}_k\proj{k}$.
Hence,
\begin{eqnarray*}
\mathbb{I}\otimes \Phi(\proj{\beta})&=&\frac{1}{d}\sum_{i,j}\ket{i}\bra{j}\otimes \Phi(\ket{i}\bra{j})\\
&=&\frac{1}{d}\sum_{i,j}\sum_k\ket{i}\bra{j}\otimes d^{(i, j)}_k\proj{k}\\
&=&\frac{1}{d}\sum_k\left(\sum_{i, j}d^{(i, j)}_k\ket{i}\bra{j}\right)\otimes \proj{k}.
\end{eqnarray*}
Since $\Phi$ is a quantum channel (CPTP map), from Choi-Jamio{\l}kowski isomorphism $\mathbb{I}\otimes \Phi(\proj{\beta})$ is positive semi-definite. Therefore, $\mathbb{I}\otimes \proj{k}[\mathbb{I}\otimes \Phi(\proj{\beta})]\mathbb{I}\otimes \proj{k}$ is a positive operator, i.e., $\sum_{i, j}d^{(i, j)}_k\ket{i}\bra{j}$ is positive for any $k$. Then $\mathbb{I}\otimes \Phi(\proj{\beta})$ can be written as $\sum_k\lambda_k \rho_k\otimes \proj{k}$, i.e., $\mathbb{I}\otimes \Phi(\proj{\beta})$ is a separable state. Therefore, $\Phi$ is an entanglement breaking channel. Naturally, since entanglement is a form of coherence \cite{Alex15}, then if a quantum channel is coherence breaking, it must be entanglement breaking too and the above calculation shows this explicitly.

Recall that a quantum channel is entanglement breaking if and only if it can be written as the Holevo
form (see Eq. \eqref{eq:Hov}) \cite{Horodecki2003}. Therefore, $\Phi$ can be written as
\begin{eqnarray*}
\Phi(\rho)=\sum_k Q_k \mathrm{Tr}(\rho P_k).
\end{eqnarray*}
As $\Phi(\rho)$ is an incoherent state for any state $\rho$, then $\Phi(\rho)$ is a diagonal state. Thus,
it can also be written as
\begin{eqnarray*}
\Phi(\rho)&=&\sum_{k,i} \bra{i}Q_k\ket{i} \mathrm{Tr}(\rho P_k)\proj{i}\\
&=&\sum_i\proj{i}(\sum_k\bra{i}Q_k\ket{i} \mathrm{Tr}(\rho P_k))\\
&=&\sum_i\proj{i}(\mathrm{Tr}\rho\otimes\proj{i}(\sum_k P_k\otimes Q_k))\\
&=&\sum_i\proj{i}\mathrm{Tr}(\rho F_i),
\end{eqnarray*}
where $F_i=\mathrm{Tr}_2 \left(\left(I\otimes\proj{i}\right)\left(\sum_k P_k\otimes Q_k\right)\right)$ and $\mathrm{Tr}_2$ denotes the partial trace on the second system. Since $\Phi$ is a CPTP map, $\sum_iF_i=\mathbb{I}$.

The implication $(iv)\Rightarrow (i)$: As each $F_i$ can be written as $F_i=\sum_k\lambda_{ik}\proj{\phi^i_{k}}$, $\Phi$ has the  Kraus representation $\set{K_{ik}}$ with $K_{ik}=\sqrt{\lambda_{ik}}\ket{i}\bra{\phi^i_{k}}$. Then, $\Phi$ is an incoherent operation and for any state $\rho$, $\Phi(\rho)$ is incoherent. Moreover, $K_{ik}\rho K^\dag_{ik}$ is incoherent, which means $\mathcal{S}_{\mathrm{cbc}}\subset \mathcal{S}_{\mathrm{scbc}}$.
\end{proof}

Let us compare the characterizations of these two kinds of coherence breaking channels. It is amazing that Theorems $\ref{th:scbt}$ and $\ref{thm:cbt}$ show that  $\mathcal{S}_{\mathrm{scbc}}=\mathcal{S}_{\mathrm{cbc}}$. Moreover, $\mathcal{S}_{\mathrm{scbc}}=\mathcal{S}_{\mathrm{cbc}}\subsetneq \mathcal{S}_{\mathrm{qc}} \subsetneq \mathcal{S}_{\mathrm{ebt}}$ (see also Fig.\ref{fig1}). Here $\mathcal{S}_{\mathrm{qc}}$ is the set of quantum-classical channels and $\mathcal{S}_{\mathrm{ebt}}$ is the set of entanglement breaking channels.

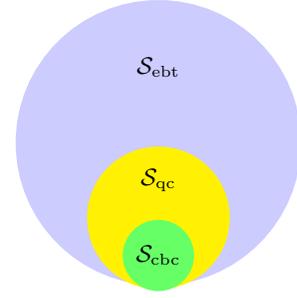
\begin{figure}
  \centerline{
    \begin{tikzpicture}[thick]
    \tikzstyle{operator} = [draw,fill=white,minimum size=1.5em]
   \tikzstyle{phase1} = [fill=blue!20,shape=circle,minimum size=12em,inner sep=0pt]
    \tikzstyle{phase2} = [fill=yellow,shape=circle,minimum size=6em,inner sep=0pt]
    \tikzstyle{phase3} = [fill=green!60,shape=circle,minimum size=3em,inner sep=0pt]
    \tikzstyle{surround} = [fill=blue!10,thick,draw=black,rounded corners=2mm]
    \tikzstyle{block} = [rectangle, draw, fill=white,
    text width=3em, text centered, , minimum height=6em]
    \node[phase1] (p1) at (5,3) {~};
    \node[phase2] (p2) at (5,2) {~};
    \node[phase3] (p3) at (5,1.5) {~};
    \node at (5,4) (e1){$\mathcal{S}_{\mathrm{ebt}}$};
    \node at (5,2.5) (e2){$\mathcal{S}_{\mathrm{qc}}$};
  \node at (5,1.5) (e3){$\mathcal{S}_{\mathrm{cbc}}$};
    \end{tikzpicture}
  }
  \caption{The relationship between the sets of entanglement breaking channels (EBTs), quantum-classical channels (QCs) and coherence breaking channels (CBCs).}
  \label{fig1}
\end{figure}

\section{Inter-relations of SIO\texorpdfstring{\MakeLowercase{s}}{}, DIO\texorpdfstring{\MakeLowercase{s}}{} and CBC\texorpdfstring{\MakeLowercase{s}}{}}
\label{sec:rel}
A special kind of incoherent operation called strictly incoherent operation (SIO), has been proposed recently \cite{Chitambar2016}. Any quantum operation $\Phi$ is a SIO if and only if  it can be represented by a set of Kraus operators $\set{M_i}$ with $M_i=\sum^{d-1}_{j=0} d_{ij}\ket{\pi_i(j)}\bra{j}$ \cite{Chitambar2016}. Let the set of all SIOs be denoted by $\mathcal{S}_{\mathrm{sio}}$.

\begin{prop}
\label{prop:sio}
Any quantum operation $\Phi$ belongs to both $\mathcal{S}_{\mathrm{sio}}$ and $\mathcal{S}_{\mathrm{cbc}}$ if and only if it can be represented by Kraus operators $K_{ij}$ of the form: $K_{ij}=d_{ij}\ket{\pi_i(j)}\bra{j}$, where $\pi_i$ is a permutation and $j\in\set{0,\ldots,d-1}$.
\end{prop}

\begin{proof}
The sufficiency part follows directly from the definition of $\mathcal{S}_{\mathrm{sio}}$ and the characterization of CBCs as given in Theorem \ref{thm:cbt}. We only need to prove the necessary part of the proposition. Since any SIO can be represented by Kraus operators $\set{M_i}$ with $M_i=\sum^{d-1}_{j=0} d_{ij}\ket{\pi_i(j)}\bra{j}$. Thus
\begin{eqnarray*}
\Phi(\rho)&=&\sum_iM_i\rho M^\dag_i\\
&=&\sum_i \Pa{\sum_jd_{ij}\ket{\pi_i(j)}\bra{j}}
\rho \Pa{\sum_k\bar{d}_{ik}\ket{k}\bra{\pi_i(k)}}\\
&=&\sum_{j,k}\rho_{jk}\sum_id_{ij}\bar{d}_{ik}
\ket{\pi_i(j)}\bra{\pi_i(k)}
\end{eqnarray*}
Besides, $\Phi$ is also a coherence breaking operation. Therefore, $\Phi(\proj{\psi})$ is an incoherent state for any maximally coherent state $\ket{\psi}=\frac{1}{\sqrt{d}}\sum_{j=0}^{d-1}e^{i\theta_j} \ket{j}$, which implies
\begin{align}\label{eq:inco1}
\Phi(\proj{\psi})&=\frac{1}{d}\sum_{j,k}e^{i(\theta_j-\theta_k)}\sum_id_{ij}\bar{d}_{ik}\ket{\pi_i(j)}\bra{\pi_i(k)}\nonumber\\
&=\frac{1}{d}\sum_{j,k}e^{i(\theta_j-\theta_k)}P^{(j,k)}\in \mathcal{I},
\end{align}
where $P^{(j,k)}:=\sum_id_{ij}\bar{d}_{ik} \ket{\pi_i(j)}\bra{\pi_i(k)}$. Then it is easy to see that $P^{(k,k)}$ is diagonal but the diagonal part of $P^{(j,k)}$ is always zero for any $j\neq k$, i.e., $P^{(j,k)}_{rr}=0$ for any $r\in \set{0,\ldots,d-1}$. Moreover, Eq.\eqref{eq:inco1} means  that for any fixed $r,s\in \set{0,\ldots,d-1}$ with $r\neq s$, $\sum_{j,k}e^{i(\theta_j-\theta_k)}P^{(j,k)}_{rs}=0$ for any $\set{\theta_i}^{d-1}_{i=0}$ with $\theta_i\in \real$, where $P^{(k,k)}_{rs}=0$ as $P^{(k,k)}$ is diagonal. By Lemma \ref{lem:equ0} (see \cref{append:lemmas}), we have $P^{(j,k)}_{rs}=0$ for any $j,k$ when $r\neq s$. Hence $P^{(j,k)}=0$ for any  $j\neq k$. Therefore,
\begin{eqnarray*}
\Phi(\rho)&=&\sum_{j,k}\rho_{jk}P^{(j,k)} =\sum_{k}\rho_{kk}P^{(k,k)}\\
&=&\sum_{k}\rho_{kk}\sum_i|d_{ik}|^2\proj{\pi_i(k)}\\
&=&\sum_{ij}K_{ij}\rho K^\dag_{ij},
\end{eqnarray*}
with $K_{ij}=d_{ij}\ket{\pi_i(j)}\bra{j}.$
\end{proof}

The above proposition shows that the operations that belong to both $\mathcal{S}_{\mathrm{cbc}}$ and $\mathcal{S}_{\mathrm{sio}}$ are trivial. These operations have action on any state (up to permutations) as follows: $\Delta(\cdot)=\sum_i\bra{i}\cdot\ket{i}\ket{i}\bra{i}$. Thus, we can use the difference between CBCs and SIOs to show the difference between IOs and SIOs, especially in some operational task.
It has been proved that in the task of incoherent teleportation \cite{StreltsovA2015}, it is possible to implement the perfect incoherent teleportation of an unknown state of one qubit with the help of one singlet and two bits of classical communication. But if we restrict the IOs to the SIOs, such a task may be unachievable.
\begin{prop}
Perfect strictly incoherent teleportation of an unknown state of one qubit is not possible with one singlet and two bits of classical communication.
\end{prop}
\begin{proof}
Recall that in the standard teleportation protocol, the initial state of the whole system is given by
\begin{eqnarray*}
\ket{\Psi}=\ket{\gamma}_{A'}\ot \ket{\phi^+}_{AB},
\end{eqnarray*}
where $\ket{\phi^+}_{AB}=(\ket{00}+\ket{11})/\sqrt{2}$ is a maximally entangled state and $\ket{\gamma}_{A'}=\alpha\ket{0}+\beta\ket{1}$ is the unknown state.
Then we show that the perfect teleportation is not possible by the strictly incoherent measurements on Alice's and Bob's systems and the classical communication between them. For any strictly incoherent measurement $K_A=\sum^1_{i,j=0}d_{ij}\ket{\pi(ij)}\bra{ij}$ on Alice's system with $\pi$ being a permutation, then after such a measurement the reduced state (unnormalized) on Bob's system is given by
\begin{eqnarray*}
\rho_B=(\abs{\alpha d_{00}}^2+\abs{\beta d_{10}}^2)\proj{0}+(\abs{\alpha d_{01}}^2+\abs{\beta d_{11}}^2)\proj{1}.
\end{eqnarray*}
It is obvious that $\rho_B$ can not be transformed to $\proj{\gamma}$ by strictly incoherent operations since coherence cannot increase under incoherent operations.
\end{proof}

Also, another set of operations has been proposed while pursuing for an operationally meaningful resource theory of coherence. These are called dephasing covariant incoherent operations (DIOs) \cite{Chitambar2016}. An operation $\Phi$ is called a dephasing-covariant incoherent operation if
\begin{eqnarray}
\Br{\Delta,\Phi}=0,
\end{eqnarray}
where  $\Delta(\rho):=\sum_{i} \bra{i}\rho \ket{i}\ket{i}\bra{i}$. Let the set of all DIOs be denoted by $\mathcal{S}_{\mathrm{dio}}$. Next, we find the relation between DIOs and CBCs as follows.

\begin{prop}
\label{prop:dio}
Any quantum operation $\Phi$ belongs to both $\mathcal{S}_{\mathrm{dio}}$ and $\mathcal{S}_{\mathrm{cbc}}$ if and only if it can be represented by Kraus operators $K_{ij}$ of the form: $K_{ij}=\sqrt{p_{ij}}\ket{i}\bra{j}$, where $i,j\in\set{0,\ldots,d-1}$, $p_{ij}\geq0$ and $\sum^{d-1}_{i=0}p_{ij}=1$.
\end{prop}

\begin{proof}
The sufficiency part follows directly from the definition of DIOs and Theorem \ref{thm:cbt}. The necessary part of the proposition can be proved as follows. We know that any CBC $\Phi$ can be expressed as
\begin{eqnarray*}
\Phi(\rho)=\sum_{i}\proj{i}\mathrm{Tr}(\rho F_i),
\end{eqnarray*}
where $\set{F_i}$ is a set of positive semi-definite operators and $\sum_i F_i=\mathbb{I}$. Moreover, $\Phi\in \mathcal{S}_{\mathrm{dio}}$ implies that
$\Delta(\Phi(\ket{j}\bra{k}))=0$  for any $j\neq k$  \cite{Chitambar2016} and we have
$\mathrm{Tr}(\ket{j}\bra{k} F_i)=0$ for any $j\neq k$, which means that $F_i$ is diagonal for any $i$ and can be written as $F_i=\sum^{d-1}_{j=0}p_{ij}\ket{j}\bra{j}$ with $p_{ij}\geq0$. Since $\sum_iF_i=\mathbb{I}$, then $\sum_ip_{ij}=1$ for any $j$. Therefore, such $\Phi$ can be expressed by the Kraus operators of the form $K_{ij}=\sqrt{p_{ij}}\ket{i}\bra{j}$.
\end{proof}
Moreover, from Propositions \ref{prop:sio} and \ref{prop:dio}, it is easy to see that $\mathcal{S}_{\mathrm{sio}}\cap \mathcal{S}_{\mathrm{cbc}}=\mathcal{S}_{\mathrm{dio}}\cap \mathcal{S}_{\mathrm{cbc}}$.

\section{Coherence breaking indices}\label{sec:dis}
In this section we discuss the iterative behaviour of quantum channels on the system of interest. In particular, we elaborate on how many iterations of a given incoherent quantum channel are needed in order for it to completely destroy the coherence of any input state or turn it into a coherence breaking channel? The minimum number of iterations of a given incoherent quantum channel is termed as {\it coherence breaking index} of the same channel. The coherence breaking indices of incoherent quantum channels can be considered as their relative figure of merit in terms of their decohering powers. There naturally appear quantum systems in various practical scenarios whose noise can be considered as a single elementary process iterated step by step in time. With an experimentally well-grounded assumption where these elementary steps are completely independent with each other, the action of the noise becomes a stroboscopic Markov process and can be modelled by an $n$-fold iteration of a given quantum channel \cite{Lami2015}. This justifies the consideration of coherence breaking indices from an experimental viewpoint.

Let $\Phi$ be an incoherent quantum channel. The {\it coherence breaking index} $n(\Phi)$ of $\Phi$ is defined as
\begin{eqnarray}
n(\Phi)&=&\min\set{n\geq 1: \Phi^n\in \mathcal{S}_{\mathrm{cbc}}}.
\end{eqnarray}
It is easy to see that if $U$ is an incoherent unitary operator, then
\begin{eqnarray*}
n(U\Phi U^\dag)=n(\Phi).
\end{eqnarray*}
For a quantum channel $\Phi$, if $n(\Phi)=\infty$, i.e., for any finite $n$, $\Phi^n$ is not a coherence breaking channel, one can term it as a coherence saving channel as it never destroys coherence completely. In the following we consider some examples of incoherent quantum channels and calculate their coherence breaking indices.

\smallskip
\noindent
\emph{Example 1.}
Consider an incoherent qubit quantum channel $\Phi$ characterized by $(M,\vec{n})$ with
\begin{equation*}
M= \Br{\begin{array}{ccc}
0 & \alpha & 0  \\
0 & 0 & 0\\
0 & 0 & 0
\end{array}},
\end{equation*}
where $\alpha$ is a real number and $\vec{n}=(0, 0, 0)^T$. For $\abs{\alpha}\leq 1$, $\Phi$ represented by $(M,\vec{n})$ is a CPTP map \cite{Ruskai2002, King2001}. The channel $(M,\vec{n})$ is then an incoherent channel but not a coherence breaking channel. Note that if a qubit quantum channel $\Phi$ is characterized by $(M,\vec{n})$, then iterated channel $\Phi^n$ is characterized by $\left(M^n,(\sum_{k=0}^{n-1}M^k)\vec{n}\right)$. Thus, it is easy to see that $\Phi^2$ is a coherence breaking channel and hence $n(\Phi)=2$. Let us consider another less trivial example of an incoherent qubit quantum channel $\Phi$ characterized by $(M,\vec{n})$ with
\begin{equation*}
M= \Br{\begin{array}{ccc}
0 & \alpha & 0  \\
0 & 0 & 0\\
\beta & 0 & 0
\end{array}},
\end{equation*}
where $\alpha$ and $\beta$ are real numbers and $\vec{n}=(0, 0, n_z)^T$. Again, we can choose $\alpha, \beta, n_z$ appropriately such that $(M,\vec{n})$ is a CPTP map.  The channel $(M,\vec{n})$ is then an incoherent channel but not a coherence breaking channel. Again, $n(\Phi)=2$.

\smallskip
\noindent
\emph{Example 2.}
Consider generalized amplitude damping channels \cite{Lami2015} on qubit systems as given by $D_{p,t}[a\ket{0}\bra{0} +  b\ket{0}\bra{1}+  b^*\ket{1}\bra{0} + c\ket{1}\bra{1}]= [pa+t(1-p)(a+c)]\ket{0}\bra{0} +  \sqrt{p}b\ket{0}\bra{1}+  \sqrt{p}b^*\ket{1}\bra{0} + [-pa+(1-t+pt)(a+c)]\ket{1}\bra{1}$.
The representation $(M,\vec{n})$ of the qubit channel $D_{p,t}$ is given by
\begin{equation}
M= \Br{\begin{array}{ccc}
\sqrt{p} & 0 & 0  \\
0 & \sqrt{p} & 0\\
0 & 0 & p
\end{array}}\mathrm{~and~}
\vec{n}=(1-p)(2t-1)\Br{\begin{array}{ccc}
0 \\
0 \\
1
\end{array}}.
\end{equation}
Moreover, the representation of $D^n_{p,t}$ is given by $(\tilde{M},\vec{\tilde{n}})$, where
\begin{align*}
&\tilde{M}= \Br{\begin{array}{ccc}
\sqrt{p^n} & 0 & 0  \\
0 & \sqrt{p^n} & 0\\
0 & 0 & p^n
\end{array}}\mathrm{~and~}\\
&\vec{\tilde{n}}=\sum_{k=0}^{n-1}p^k\vec{n}=(1-p^n)(2t-1)\Br{\begin{array}{ccc}
0 \\
0 \\
1
\end{array}}.
\end{align*}
Thus, we have $D^n_{p,t} = D_{p^n,t}$. It means that the coherence breaking index $n(D_{p,t})$ of $D_{p,t}$ is not finite.

\section{ Coherence sudden death and universality of the dynamics of coherence}
\label{sec:coh-sud}
Consider a dynamical evolution of a single quantum system in a state $\ket{\psi}$ under some quantum channel $\Phi$. The phenomenon of vanishing of coherence of $\ket{\psi}$ in some finite time is termed as sudden death of coherence. If the coherence of $\ket{\psi}$ does not vanish in some finite time or vanishes asymptotically then this phenomenon is termed as no sudden death of coherence. Moreover, in the case stroboscopic Markovian processes where the evolution of a quantum system in a state $\ket{\psi}$ is modelled by an $n$-fold iterations of an elementary channel $\Phi$, if the coherence of $\ket{\psi}$ vanishes in $n_0$ iterations of $\Phi$ with $n_0<n$, then we say that such stroboscopic Markovian processes lead to the coherence sudden death. It is important to note that the phenomenon of coherence sudden death is both initial state and channel dependent. However, in the case of qubit states and for a specific measure of coherence, namely the $l_1$ norm of coherence \cite{Baumgratz2014}, we show that coherence sudden death is only channel dependent irrespective of the initial state. To achieve this we first state the factorization relation for the $l_1$ norm of coherence obtained in the evolution equation of coherence \cite{Hu2016}.

In a $d$-dimensional Hilbert space, any quantum state can be represented as
\begin{align}
 \rho=\frac{1}{d}\mathbb{I}_d + \frac{1}{2}\vec{x}\cdot\vec{\Lambda},
\end{align}
where $\vec{x}=(x_1,\ldots,x_{d^2-1})$, $\vec{\Lambda}=(\Lambda_1,\ldots,\Lambda_{d^2-1})$, $x_i=\mathrm{Tr}[\rho \Lambda_i]$ with $\Lambda_i$ being the generators of  $su(d)$ \cite{Hioe1981,Mahler98,Byrd2003,Kimura03}. The vector $\vec{x}$ can be written as $\vec{x}=\chi \vec{n}$, where $\vec{n}=(n_1,\ldots,n_{d^2-1})$ is a unit vector in $\mathbb{R}^{d^2-1}$ and $|\chi|\leq \sqrt{2(d-1)/d}$. Now the factorization relation of Ref. \cite{Hu2016} can be stated as follows.

\begin{lem}[\cite{Hu2016}]\label{lem:Hu}
Let us consider a quantum operation $\Phi$ with $\Phi\left(\frac{\mathbb{I}}{d}\right)$ is diagonal, then for any quantum state $ \rho=\frac{1}{d}\mathbb{I}_d + \frac{1}{2}\chi\vec{n}\cdot\vec{X}$
\begin{align}
 C_{l_1}\left( \Phi(\rho) \right) = C_{l_1}\left( \rho \right) C_{l_1}\left( \Phi(\rho_P) \right),
\end{align}
where $ \rho_P=\frac{1}{d}\mathbb{I}_d + \frac{1}{2}\chi_P \vec{n}\cdot\vec{X}$ is called the probe state and $\chi_P = 1/\sum_{r=1}^{(d^2-d)/2}(n_{2r-1}^2+n_{2r}^2)^{1/2}$.
\end{lem}
It is easy to see that incoherent operations satisfy the condition of the above Lemma, therefore, for incoherent operations, the above equality holds. Moreover, the above lemma can be simplified for the qubit cases and we have the following proposition.
\begin{prop}
\label{prop-Hu}
If $\Phi$ is an incoherent operation on a single qubit system, then for any qubit state $\rho$ there exists a maximally coherent state $\ket{\psi}$ such that
\begin{align}
 C_{l_1}\left( \Phi(\rho) \right) = C_{l_1}\left( \rho \right) C_{l_1}\left( \Phi(\ket{\psi}\bra{\psi}) \right).
\end{align}
\end{prop}
\begin{proof}
Since any qubit state $\rho$ can be written as $\rho=\frac{1}{2}[(1+z)\ket{0}\bra{0} + (x+iy)\ket{0}\bra{1}+(x-iy)\ket{1}\bra{0}+(1-z)\ket{1}\bra{1}]$ and $\Phi$ is an incoherent operation, then $C_{l_1}\left( \rho \right) = |x+iy|$ and
\begin{align}
  &C_{l_1}\left( \Phi(\rho) \right)\nonumber\\
  &= 2\left| \frac{x+iy}{2}\bra{1}\Phi(\ket{0}\bra{1})\ket{0} +  \frac{x-iy}{2}\bra{1}\Phi(\ket{1}\bra{0})\ket{0} \right|\nonumber\\
  &=|x+iy|\left|e^{i\theta} \bra{1}\Phi(\ket{0}\bra{1})\ket{0} + e^{-i\theta}\bra{1}\Phi(\ket{1}\bra{0})\ket{0} \right|,
\end{align}
where $x+iy=|x+iy|e^{i\theta}$ and $x-iy=|x+iy|e^{-i\theta}$. Now, taking $\ket{\psi}=\frac{1}{\sqrt{2}}(\ket{0}+e^{-i\theta}\ket{1})$ as the maximally coherent state, we have
\begin{align}
\label{eq:factor}
 C_{l_1}\left( \Phi(\rho) \right) = C_{l_1}\left( \rho \right) C_{l_1}\left( \Phi(\ket{\psi}\bra{\psi}) \right).
\end{align}
This concludes the proof of the proposition.
\end{proof}

\begin{figure}
\includegraphics[width=75mm]{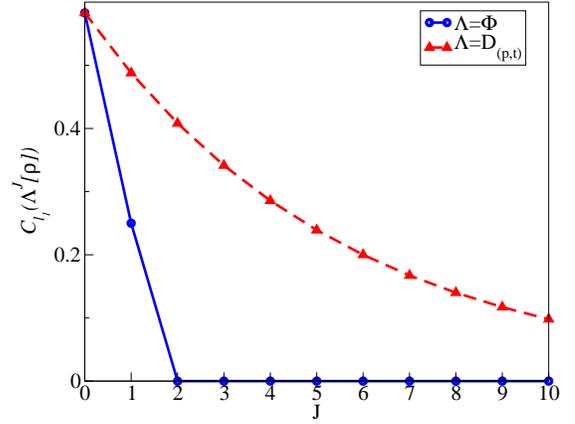}
\caption{Consider a qubit state $\rho = \frac{1}{2}\left( \mathbb{I}+\vec{r}\cdot\vec{\sigma} \right)$ with $\vec{r} = (0.3,0.5,0.2)^T$. Also, take $J=10$ in the stroboscopic Markovian processs of both the examples $3$ and $4$. In example $3$, take $\alpha=0.5$.  We have $C_{l_1}(\rho) = 0.5830$, $C_{l_1}(\Phi[\rho]) = 0.25$ and $C_{l_1}(\Phi^J[\rho]) = 0$ for $J\geq2$. The solid line in above figure plots $C_{l_1}(\Phi^J[\rho])$ as function of $J$ and shows the coherence sudden death at the second iteration of the channel. In example $4$, take $p=0.7$ and $t=1$. Here $C_{l_1}(D_{p,t}^J[\rho]) = p^{\frac{J}{2}}C_{l_1}(\rho)$. The dotted line in above figure plots $C_{l_1}(D_{p,t}^J[\rho])$ as function of $J$ and shows no sudden death of coherence as the coherence is nonzero at every iteration of the channel $D_{p,t}$.}
\label{fig:sudden-death}
\end{figure}

The above proposition implies that the knowledge of the initial coherence of a quantum system and the action of incoherent operation $\Phi$ on maximally coherent state are enough to determine the evolution of coherence. However, we remark that the maximally coherent state appearing in the above proposition is initial state dependent. The phenomenon of coherence sudden death in these cases is independent of the initial state of quantum system and is essentially endowed to the channel. It is important to note that if a channel $\Phi$ is coherence breaking which from Theorem \ref{thm:cbt} means that $C_{l_1}\left( \Phi(\ket{\psi}\bra{\psi}) \right)=0$, it necessarily implies coherence sudden death and vice-versa for any initial state.

For a stroboscopic Markovian process represented by $\Phi^J$, i.e., by $J$ iterations of the channel $\Phi$, Eq. (\ref{eq:factor}) becomes
\begin{align}
\label{eq:strob}
 C_{l_1}\left( \Phi^J(\rho) \right) = C_{l_1}\left( \rho \right) C_{l_1}\left( \Phi^J(\ket{\psi}\bra{\psi}) \right).
\end{align}
The above equation can be reinterpreted as follows. If coherence breaking index of a channel $\Phi$ is $n(\Phi)$, then after $n(\Phi)$ iterations of the channel $\Phi$, the channel $\Phi$ becomes a coherence breaking channel. This implies that for a channel $\Phi$ with coherence breaking index $n(\Phi)$, in the corresponding stroboscopic Markovian process with $J$ iterations of the channel $\Phi$, the evolution will lead to the coherence sudden death if $n(\Phi)<J$. Next, we present a few examples of incoherent evolutions that can lead to the coherence sudden death.

\smallskip
\noindent
\emph{Example 3.}
Consider a stroboscopic Markovian process given by $\Phi^J$ with $J\geq3$, where $\Phi$ is a qubit quantum channel characterized by $(M,\vec{n})$ with
\begin{equation*}
M= \Br{\begin{array}{ccc}
0 & \alpha & 0  \\
0 & 0 & 0\\
0 & 0 & 0
\end{array}},
\end{equation*}
and $\vec{n}=(0, 0, 0)^T$. Here $\alpha ~(\abs{\alpha}\leq 1)$ is a real number. We obtained earlier that the channel $\Phi$ has coherence breaking index $n(\Phi)=2$. Thus, the above stroboscopic Markovian process leads to the coherence sudden death for any input state (see also Fig. \ref{fig:sudden-death}).

\smallskip
\noindent
\emph{Example 4.}
Consider a stroboscopic Markovian process given by $D_{p,t}^J$, where $D_{p,t}$ is a generalized amplitude damping qubit channel (see Sec. \ref{sec:dis}). The coherence breaking index of the generalized amplitude damping $n(D_{p,t})$ is not finite, therefore, such a stroboscopic Markovian process will never lead to coherence sudden death (see also Fig. \ref{fig:sudden-death}).

\smallskip
\noindent
\emph{Typicality of evolution of coherence for higher dimensional systems.--} It is known that pure quantum states in higher dimensional Hilbert spaces show concentration of measure phenomenon for various physical properties (for example, see Refs. \cite{Hayden2006, Popescu2006, Hayden2008, Tiersch2009, Datta2010, Muller2011, Tiersch2013, UttamS2015, Zhang2015, BuD2016}). Similarly, here we aim at finding the typical properties of the evolution of coherence under incoherent quantum channels. Consider a quantum system in a state $\rho(0)=\ket{\psi}\bra{\psi}$ evolving under an incoherent quantum channel $\Lambda_t$ such that $\rho(t) = \Lambda_t[\rho(0)]$. Specifically, we want to describe the dynamics of pure states $\ket{\psi}$ chosen uniformly at random from the Haar measure under some quantum channel $\Lambda_t$ and wish to establish the typicality of the evolution. For this we will resort on the concentration of measure phenomenon encapsulated in L\'evy's lemma applicable to Lipschitz continuous functions (see \cref{levy's lemma}).
Now, consider two pure states $\ket{\psi}$ and $\ket{\phi}$ evolving under $\Lambda_t$. Then, for any Lipschitz continuous measure $C$ of coherence with Lipschitz constant $\eta_C$, we have
\begin{align}
&\left|C\left(\Lambda_t[\proj{\psi}]\right)-C\left(\Lambda_t[\proj{\phi}]\right)\right|\nonumber\\
&~~~~~~~~~~~~~~\leq \eta_C \left|\left|~\Lambda_t[\proj{\psi}]-\Lambda_t[\proj{\phi}]~\right |\right|_1\nonumber\\
&~~~~~~~~~~~~~~\leq\eta_C \eta_{\Lambda_t} \left|\left|~\proj{\psi}-\proj{\phi}~\right |\right|_1\nonumber\\
&~~~~~~~~~~~~~~\leq 2\eta_C \eta_{\Lambda_t} \left|\left|~\ket{\psi}-\ket{\phi}~\right |\right|,
\end{align}
where $||\cdot||_1$ is the trace norm and $||\cdot||$ is the Euclidean norm (see \cref{levy's lemma}). The first inequality follows from the definition of the Lipschitz continuous function over the space of density matrices. The second inequality follows from the monotonicity of the trace norm under quantum channels \cite{Ruskai1994} (see also \cref{levy's lemma}). The third inequality follows from the relation between the trace distance and the Euclidean distance \cite{Tiersch2013}. If the distance between the states $\ket{\psi}$ and $\ket{\phi}$ is small then the coherences of the final states are also close to each other. Now, based on L\'evy's lemma we have following result.
\begin{thm}
\label{th:generic}
Let $\ket{\psi}$ be a random pure state on a $d$ dimensional Hilbert space. Then, for any $\epsilon\geq 0$, we have
\begin{align}
\label{Eq:conc}
\mathrm{Pr} &\left\{ \left|C(\Lambda_t[\proj{\psi}]) -\mathbb{E}_\psi C\right|> \epsilon \right\} \nonumber\\
&~~~~~~~~~~~~~~~\leq 2\exp\left( -\frac{d \epsilon^2}{18 \pi^3\eta_C^2\eta_{\Lambda_t}^2 \ln 2 } \right),
\end{align}
where $2\eta_C\eta_{\Lambda_t}$ is the Lipschitz constant for the function $F:\mathbb{CP}^{d-1}\mapsto \mathbb{R}$ with $F(\ket{\psi}) = C(\Lambda_t[\proj{\psi}])$ and $\mathbb{E}_\psi C = \int \mathrm{d}\psi C(\Lambda_t[\proj{\psi}])$.
\end{thm}
As an example, the Lipschitz constant for the scaled $l_1$ norm of coherence, i.e., $C_{l_1}/C_{l_1}^\mathrm{max}$ over the space of the density matrices is given by $\frac{d}{(d-1)}$ (see \cref{append:lip}). Therefore, the scaled $l_1$ norm of coherence of the evolved state starting from a generic state almost always concentrates around the average value given by $\frac{\mathbb{E}_\psi C_{l_1}}{C_{l_1}^\mathrm{max}} =\frac{1}{C_{l_1}^\mathrm{max}} \int \mathrm{d}\psi C_{l_1}(\Lambda_t[\proj{\psi}])$. In particular, we have the following corollary.
 \begin{cor}
\label{cor:generic}
Let $\ket{\psi}$ be a random pure state on a $d$ dimensional Hilbert space. Then, for any $\epsilon\geq 0$, we have
\begin{align}
\label{eq:l1}
\mathrm{Pr} &\left\{ \left|\frac{C(\Lambda_t[\proj{\psi}])}{C_{l_1}^{\mathrm{max}}} -  \frac{\mathbb{E}_\psi C_{l_1}}{C_{l_1}^\mathrm{max}}  \right|> \epsilon \right\} \nonumber\\
&~~~~~~~~~~~~~~~\leq 2\exp\left( -\frac{(d-1)^2 \epsilon^2}{18 \pi^3  \eta_{\Lambda_t}^2 d\ln 2} \right).
\end{align}
\end{cor}

\section{Summary and outlook}\label{sec:con}
In this work, we have investigated quantum channels which output only incoherent states for any input states. We call such channels as coherence breaking channels. First we define two kinds of coherence breaking channels, namely, coherence breaking and selective coherence breaking channels, and obtain the full characterization of these two types of quantum channels. Then we prove that they are, in fact, equivalent. Further, we consider stroboscopic Markovian processes in which the action of noise is characterized by the iterative applications of an elementary CPTP map. In these situations, we define coherence breaking indices of incoherent quantum channels which can be considered as a relative figure of merit in deciding the detrimental capabilities of a quantum channel in the context of quantum coherence. We then define the notion of coherence sudden death under quantum channels which describes abrupt vanishing of coherence under an incoherent channel in time. Based on the recently obtained factorization relations for the $l_1$ norm of coherence, we link the coherence breaking channels and coherence breaking indices with the coherence sudden death and present various examples to delineate this. Finally, for systems with higher dimensional Hilbert spaces, based on L\'evy's lemma, we show the typicality of the dynamics of coherence for random pure states. This is a very useful result in depicting the behaviour of a quantum channel acting on a quantum system with higher dimensional Hilbert space together with the reduction of the computational complexity of coherence evolution. We exemplify this phenomenon by considering the scaled $l_1$ norm of coherence for random pure states and provide explicit bounds on the typical coherence of the evolved states.

The results in this work present a systematic and exhaustive characterization of the detrimental effects of various noisy scenarios and therefore, are of great practical value. Moreover, the results on the typicality of the dynamics of the coherence provide a tractable estimation of the dynamics of coherence in otherwise computationally hard scenarios of the systems with higher dimensional Hilbert spaces. However, more work is possible in this context. For example, in our work we left open the calculation of the average coherence for various relevant noise scenarios such as the random incoherent unitary evolution considering the scaled $l_1$ norm of coherence as a measure of coherence. It will be interesting for future works to explicitly obtain results in this direction. Also, it will be useful to obtain simplified factorization relations in the context of evolution of coherence, say for pure qudit states, in future.

\smallskip
\noindent
\begin{acknowledgments}
J.W. is supported by the Natural Science Foundations of China (Grants No.11171301 and No. 10771191) and the Doctoral Programs Foundation of the Ministry of Education of China (Grant No. J20130061). K.B. acknowledges Chunhe Xiong and Lu Li for various discussions and help. Swati thanks Arun Kumar Pati for his help and support during her stay as a visiting student at Harish Chandra Research Institute, India. U.S. acknowledges support from the research fellowship of Department of Atomic Energy, Government of India and thanks Namit Anand for various discussions related to this work.
 
\end{acknowledgments}

\appendix

\section{ Some useful Lemmas}
\label{append:lemmas}

\begin{lem}\label{lem:noze}
In a $d$-dimensional Hilbert space, for any two nozero pure states $\ket{\alpha}$ and $\ket{\beta}$, there exists a maximally coherent  pure state $\ket{\psi}$ such that $\iinner{\alpha}{\psi}\neq0$ and $\iinner{\beta}{\psi}\neq 0$.
\end{lem}
\begin{proof}
Any  maximally coherent  pure state $\ket{\psi}$ can be written as $\ket{\psi}=\frac{1}{\sqrt{d}}\sum^{d-1}_{j=0}e^{i\theta_j}\ket{j}$ and let $\ket{\alpha}=\sum^{d-1}_{j=0}|\alpha_j|e^{i a_j}\ket{j}$,
$\ket{\beta}=\sum^{d-1}_{j=0}|\beta_j|e^{ib_j}\ket{j}$. Then $\iinner{\alpha}{\psi}=\frac{1}{\sqrt{d}}\sum_j|\alpha_j|e^{i(\theta_j-a_j)}$ and $\iinner{\beta}{\psi}=\frac{1}{\sqrt{d}}\sum_j|\beta_j|e^{i(\theta_j-b_j)}$. We can use maximally coherent states to construct a basis of the Hilbert space and $\{\ket{\psi_k}\}_{k=0}^{d-1}$ forms a basis of the Hilbert space, where
\begin{eqnarray*}
\ket{\psi_k}=\frac{1}{\sqrt{d}}\sum_{j=0}^{d-1} e^{i\frac{2\pi k j}{d}}\ket{j}.
\end{eqnarray*}
Thus, there exists a maximally coherent state $\ket{\psi}=\frac{1}{\sqrt{d}}\sum^{d-1}_{j=0}e^{i\theta_j}\ket{j}$ such that $\iinner{\alpha}{\psi}\neq0$ or $\iinner{\beta}{\psi}\neq 0$. If both $\iinner{\alpha}{\psi}$ and $\iinner{\beta}{\psi}$ are nonzero, then lemma is proved. Otherwise, without loss of generality, we assume $\iinner{\alpha}{\psi}=0$ and $\iinner{\beta}{\psi}\neq 0$, i.e.,
\begin{eqnarray*}
\sum_j|\alpha_j|e^{i(\theta_j-a_j)}=0;\\
\sum_j|\beta_j|e^{i(\theta_j-b_j)}\neq 0.
\end{eqnarray*}

\smallskip
\noindent
(a) If there exist a $k$ such that $|\alpha_k|>0$ and $|\beta_k|>0$, then by the continuity, there exist $\varepsilon$ such that
\begin{eqnarray*}
\sum_j|\alpha_j|e^{i(\theta_j+\varepsilon\delta_{j,k}-a_j)}\neq0;\\
\sum_j|\beta_j|e^{i(\theta_j+\varepsilon\delta_{j,k}-b_j)}\neq 0,
\end{eqnarray*}
where $\delta_{j,k}=0$ if $j\neq k$ and $\delta_{j,k}=1$ if $j=k$. That is, $\ket{\psi'}:=\frac{1}{\sqrt{d}}\sum^{d-1}_{j=0}e^{i(\theta_j+\varepsilon\delta_{j,k})}\ket{j}$
is a maximally coherent state with $\iinner{\alpha}{\psi'}\neq0$ and $\iinner{\beta}{\psi'}\neq 0$.

\smallskip
\noindent
(b) Otherwise, $|\alpha_k\beta_k|=0$ for any $k$. Then without loss of generality, $\ket{\alpha}$ and $\ket{\beta}$ can be viewed as
\begin{eqnarray*}
\ket{\alpha}&=&\sum^{d_1}_{j=0}|\alpha_j|e^{i a_j}\ket{j};\\
\ket{\beta}&=&\sum^{d-1}_{j=d_1+1}|\beta_j|e^{i b_j}\ket{j}.
\end{eqnarray*}
Then,
\begin{eqnarray*}
\iinner{\alpha}{\psi}&=&\sum^{d_1}_{j=0}|\alpha_j|e^{i(\theta_j-a_j)}=0;\\
\iinner{\beta}{\psi}&=&\sum^{d-1}_{j=d_1+1}|\beta_j|e^{i(\theta_j-b_j)}\neq 0.
\end{eqnarray*}
As $\ket{\alpha}$ is a nozero state, then there exist a $|\alpha_k|>0$, where $0\leq k\leq d_1$.
Then, by the continuity, there exists $\varepsilon$, such that
\begin{eqnarray*}
\sum^{d_1}_{j=0}|\alpha_j|e^{i(\theta_j+\varepsilon\delta_{j,k}-a_j)}&\neq&0;\\
\sum^{d-1}_{j=d_1+1}|\beta_j|e^{i(\theta_j-b_j)}&\neq& 0.
\end{eqnarray*}
That is,  $\ket{\psi'}:=\frac{1}{\sqrt{d}}\sum^{d-1}_{j=0}e^{i(\theta_j+\varepsilon\delta_{j,k})}\ket{j}$
is a maximally coherent state with $\iinner{\alpha}{\psi'}\neq0$ and $\iinner{\beta}{\psi'}\neq 0$. This completes the proof of the lemma.
\end{proof}


\begin{lem}\label{lem:equ0}
Given $n\times n$ matrix $X=[x_{ij}]\in\complex^{n\times n}$ with $x_{ii}=0$ for $1\leq i\leq n $, if for any set $\set{\theta_i}^n_{i=1}$ with $\theta_i\in \real$,
\begin{eqnarray}\label{eq:sum0}
\sum_{i,j}e^{i(\theta_i-\theta_j)}x_{ij}=0,
\end{eqnarray}
then $X=0$.
\end{lem}
\begin{proof}
Let $\ket{\theta}=(e^{-i\theta_1}, e^{-i\theta_2},\ldots,e^{-i\theta_n})^T$ where $T$ denotes transpose, then the condition \eqref{eq:sum0} is equivalent to $\bra{\theta}X\ket{\theta}=0$ for any $\ket{\theta}$. Matrix $X$ can be written as $X=B+iC$ with Hermitian matrices $B=[b_{ij}]$ and $C=[c_{ij}]$. Moreover $B=\frac{X+X^*}{2}$ and $C=\frac{X-X^*}{2i}$
implies  $b_{ii}=0$ and $c_{ii}=0$ for any $i=1,\ldots,n$. Now $\bra{\theta}X\ket{\theta}=\bra{\theta}B\ket{\theta}+i\bra{\theta}C\ket{\theta}=0$ means
$\bra{\theta}B\ket{\theta}=\bra{\theta}C\ket{\theta}=0$, as $\bra{\theta}B\ket{\theta}$ and
$\bra{\theta}C\ket{\theta}$ are both real. The (i,j) entries of $B$ can be written as $b_{ij}=|b_{ij}|e^{i\beta_{ij}}$. Since $B$ is hermitian, $|b_{ij}|=|b_{ji}|$ and $\beta_{ij}=-\beta_{ji}$. We have
\begin{eqnarray*}
0&=&\bra{\theta}B\ket{\theta}=\sum_{i,j,i\neq j}e^{i(\theta_i-\theta_j)}b_{ij}\\
&=&\sum_{i<j}[e^{i(\theta_i-\theta_j)}b_{ij}+e^{i(\theta_j-\theta_i)}b_{ji}]\\
&=&\sum_{i<j}[e^{i(\theta_i-\theta_j+\beta_{ij})}+e^{i(\theta_j-\theta_i-\beta_{ij})}]|b_{ij}|\\
&=&2\sum_{i<j}\cos(\theta_i-\theta_j+\beta_{ij})|b_{ij}|\\
&=&2\sum_{i<j}[\cos(\theta_i-\theta_j)\cos(\beta_{ij})-\sin(\theta_i-\theta_j)\sin(\beta_{ij})]|b_{ij}|,
\end{eqnarray*}
which implies $\sum_{i<j}\cos(\theta_i-\theta_j)\cos(\beta_{ij})|b_{ij}|=\sum_{i<j}\sin(\theta_i-\theta_j)\sin(\beta_{ij})|b_{ij}|$.
Let $\eta_{ij}=\cos(\beta_{ij})|b_{ij}|$ and $\lambda_{ij}=\sin(\beta_{ij})|b_{ij}|$, then we have
\begin{eqnarray}\label{eq:tot}
\sum_{i<j}\cos(\theta_i-\theta_j)\eta_{ij}=\sum_{i<j}\sin(\theta_i-\theta_j)\lambda_{ij},
\end{eqnarray}
for any set $\set{\theta_i}^n_{i=1}$.
To solve \eqref{eq:tot}, we take special set $\set{\theta_i}^n_{i=1}$ as following. For any fixed $i$ and $j$ with $i<j$, first  let all $\theta_k=0$, then \eqref{eq:tot} becomes
 \begin{eqnarray}\label{eq:all0}
\sum_{r<s}\eta_{rs}=0.
\end{eqnarray}
Now take $\theta_i=\pi$ and $\theta_k=0$ with $k\neq i$,
\begin{eqnarray}\label{eq:1pi}
\sum_{r<s,r\neq i,s\neq i}\eta_{rs}+\sum_{u<i}(-1)\eta_{ui}+\sum_{i<v}(-1)\eta_{iv}=0.
\end{eqnarray}
Using \eqref{eq:all0} minus \eqref{eq:1pi}, we get
\begin{eqnarray}\label{eq:1pi2}
\sum_{u<i}\eta_{ui}+\sum_{i<v}\eta_{iv}=0.
\end{eqnarray}
Similarly, set $\theta_j=\pi$ and $\theta_k=0$ with $k\neq j$, we can also obtain
\begin{eqnarray}\label{eq:1pi3}
\sum_{p<j}\eta_{pj}+\sum_{j<q}\eta_{jq}=0.
\end{eqnarray}
Finally, take $\theta_i=\theta_j=\pi$ with $i<j$ and $\theta_k=0$ for any $k\neq i,j$, then \eqref{eq:tot} implies
\begin{eqnarray}\label{eq:2pi1}
&&\sum_{r<s,r,s\notin\set{i,j}}\eta_{rs}+\sum_{u<i}(-1)\eta_{ui}+\sum_{i<v,v\neq j}(-1)\eta_{iv} +\eta_{ij}\nonumber\\
&+&\sum_{p<j,p\neq i}(-1)\eta_{pj}+\sum_{j<q}(-1)\eta_{jq}=0.
\end{eqnarray}
Using \eqref{eq:all0} minus \eqref{eq:2pi1}, we have
\begin{eqnarray}\label{eq:2pi2}
\sum_{u<i}\eta_{ui}+\sum_{i<v,v\neq j}\eta_{iv}
+\sum_{p<j,p\neq i}\eta_{pj}+\sum_{j<q}\eta_{jq}=0.
\end{eqnarray}
Then,
\begin{eqnarray*}
\eqref{eq:1pi2}+\eqref{eq:1pi3}-\eqref{eq:2pi2} \Rightarrow \eta_{ij}=0, \mathrm{~for~ any}~i\mathrm{~and~}  j \mathrm{~with~} i<j.
\end{eqnarray*}
Thus \eqref{eq:tot} reduces to
\begin{eqnarray*}
\sum_{i<j}\sin(\theta_i-\theta_j)\lambda_{ij}=0,
\end{eqnarray*}
for any set $\set{\theta_i}^n_{i=1}$.
Similarly, we can get $\lambda_{ij}=0$ via choosing a special set $\set{\theta_i}^n_{i=1}$. Therefore,  $|b_{ij}|^2=\eta^2_{ij}+\lambda^2_{ij}=0$ for any $i<j$ which means that $B=0$. Using the same method, it is easy to obtain $C=0$. Now, since $X=B+iC$, we have $X=0$. This completes the proof of the lemma.
\end{proof}

\section{ Random pure states, concentration of measure phenomenon and the measures of distance on Hilbert space} \label{levy's lemma}

\smallskip
\noindent {\it  Random pure states:}
The set of pure states on a $d$-dimensional Hilbert space is a complex projective space $\mathbb{C}P^{d-1}$. This set is endowed with a unique measure $\dif(\psi)$ induced by the Haar measure $\dif\mu(U)$ on the unitary group $\mathrm{U}(d)$ \cite{Wootters1990, Zyczkowski1994, Zyczkowski2001, Bengtsson2008, Aubrun2014}. Thus, any random pure state $\ket{\psi}$ can be generated by applying a random unitary matrix $U\in \mathrm{U}(d)$ on a fixed pure state $\ket{\psi_0}$, i.e., $\ket{\psi} = U\ket{\psi_0}$. Now for any function $F$ of pure state, we have
\begin{align}
\mathbb{E}_\psi F(\psi) :=\int \dif(\psi) ~F(\psi) = \int_{\mathrm{U}(d)} \dif\mu(U)~ F(U\psi_0).\nonumber
\end{align}

\smallskip
\noindent {\it  Concentration of measure phenomenon:}
The observation that an overwhelming majority of vectors of a vector space take a fixed value for many functions defined over the vector space as the dimension of the vector space goes to infinity, is referred to as the concentration of measure phenomenon. In particular,  L\'evy's lemma is the rigorous statement about the concentration of measure phenomenon \cite{Ledoux2005} for Lipschitz continuous functions on the sphere.  We will state L\'evy's lemma shortly but before that we define Lipschitz continuous functions. Consider two metric spaces $(V_1, d_1)$ and $(V_2, d_2)$ and a function $F : V_1 \rightarrow V_2$. If there exists a real number $\eta_F$ such that $d_2(F(u),F(v)) \leq \eta_F d_1(u, v)$ for all $u, v \in V_1$, then $F$ is called a Lipschitz continuous function on $V_1$ with the Lipschitz constant $\eta_F$  \cite{Searcoid2007}.

\smallskip
\noindent
{\it L\'evy's lemma (see \cite{Ledoux2005} and \cite{Hayden2006}).--}
Let $F:\mathbb{S}^k\to \mathbb{R}$ be a Lipschitz continuous function with Lipschitz constant $\eta_F$. Here $\mathbb{S}^k$ is the $k$-sphere and $\mathbb{R}$ is the real line. Let us consider a random vector $u\in\mathbb{S}^k$. Then for any $\epsilon>0$,
\begin{align}
\label{eq:levy-lemma}
\mathrm{Pr}\set{|F(u)-\mathbb{E}_uF(u)|>\epsilon} \leq 2\exp\Pa{-\frac{(k+1)\epsilon^2}{9\pi^3\eta_F^2\ln2}},
\end{align}
where $\mathbb{E}_uF(u)$ is the expected value of $F(u)$ over random vectors $u\in\mathbb{S}^k$.

\smallskip
\noindent
{\it Trace distance.--} The trace distance between two quantum states $\rho$ and $\sigma$ is defined as \cite{Wilde13}
\begin{eqnarray}
\label{tr-dis}
||\rho-\sigma||_1:=\mathrm{Tr}\left[\sqrt{(\rho-\sigma)^2}\right].
\end{eqnarray}
The trace distance satisfies the monotonicity property under the influence of a quantum channel $\Lambda_t$ \cite{Ruskai1994}. More precisely,
\begin{align}
||\Lambda_t[\rho]-\Lambda_t[\sigma]||_1\leq \eta_{\Lambda_t} ||\rho-\sigma||_1,
\end{align}
where $\eta_{\Lambda_t}\leq 1$.

\section{The Lipschitz constant for the $l_1$ norm of coherence}
\label{append:lip}
We need to find a constant $\eta_{C_{l_1}}$ for the $l_1$ norm of coherence such that\begin{eqnarray*}
\left|C_{l_1}(\rho)-C_{l_1}(\sigma)\right|
\leq \eta_{C_{l_1}}\norm{\rho-\sigma}_1.
\end{eqnarray*}
From the definition of the $l_1$ norm of coherence, $C_{l_1}(\rho)=\norm{\rho}_{l_1}-1$, where $\norm{\rho}_{l_1}=\sum_{ij}\abs{\rho_{ij}}$, we have
\begin{eqnarray}
\left|C_{l_1}(\rho)-C_{l_1}(\sigma)\right|
&=&\abs{\norm{\rho}_{l_1}-\norm{\sigma}_{l_1}}\nonumber\\
&\leq& \norm{\rho-\sigma}_{l_1}\nonumber\\
&\leq& d\norm{\rho-\sigma}_2\nonumber\\
&\leq& d \norm{\rho-\sigma}_1,
\end{eqnarray}
where the first inequality follows from the triangle inequality of norm $\norm{\cdot}_{l_1}$,  the
second inequality comes from the fact that $\norm{A}_{l_1}=\sum_{ij}\abs{A_{ij}}\leq \sqrt{d^2\sum_{ij}\abs{A_{ij}}^2}=d\norm{A}_2$ and the third inequality follows from $\norm{A}_2\leq \norm{A}_1$ \cite{Watrous2011}. Note that the norm $\norm{\cdot}_2$ is defined as $\norm{A}_2=\sqrt{\mathrm{Tr} A^\dag A}$. Thus, Lipschitz constant for the scaled $l_1$ norm of coherence, i.e., $C_{l_1}/C_{l_1}^\mathrm{max}$, is given by $\frac{d}{d-1}$ and we have
\begin{eqnarray*}
\abs{\frac{C_{l_1}(\rho)}{C_{l_1}^{\mathrm{max}}}-\frac{C_{l_1}(\sigma)}{C_{l_1}^{\mathrm{max}}}}
\leq \frac{d}{d-1}\norm{\rho-\sigma}_1.
\end{eqnarray*}

\bibliographystyle{apsrev4-1}
 \bibliography{coh-break-lit}

\end{document}